\documentclass[12pt,journal,onecolumn]{IEEEtran}

\usepackage[normalem]{ulem}

\usepackage{acronym}		
\usepackage{tikz}
\usetikzlibrary{shapes,snakes}
\usepackage{amssymb}
\usepackage{empheq}
\usepackage{graphicx}
\usepackage{epstopdf}
\usepackage{url}
\usepackage{amsmath,amsthm,amssymb}
\usepackage{amsfonts} 
\usepackage{bm}
\usepackage{bbold}
\usepackage{graphicx}
\usepackage{cite}  
\usepackage{caption}
\usepackage{subfig}
\usepackage{dsfont}

\usepackage{algorithmic}
\usepackage[ruled,vlined,linesnumbered]{algorithm2e}

\newtheorem{thm}{Theorem}

 \DeclareRobustCommand*{\IEEEauthorrefmark}[1]{%
 	\raisebox{0pt}[0pt][0pt]{\textsuperscript{\footnotesize\ensuremath{#1}}}}

\newcommand{\ma}{\bm}
\newcommand{\ve}{\bm}

\newcommand{\field}[1]{\mathbb{#1}}

\newcommand{\NN}{{\field{N}}} 

\newcommand{\RN}{{\field{R}}^{+}}

\newcommand{\indication}[1]{\mathds{1}_{\{#1\}}}
\newcommand{\operator}[1]{\mathrm{#1}} 
\newcommand{\csi}{\operator{CSI}}

\acrodef{3GPP}{3rd generation partnership project}
\acrodef{5G}{fifth generation}
 \acrodef{ABS}{almost blank subframe}
    \acrodef{BS}{base station}
    \acrodef{CDF}{cumulative distribution function}
    \acrodef{CSI}{channel state information}
    \acrodef{CQI}{channel quality indicator}
    \acrodef{CAST}{channel allocation with scalable TTI}
\acrodef{DL}{downlink}
    \acrodef{DUDe}{downlink and uplink decoupling}
\acrodef{eICIC}{enhanced intercell interference coordination}
\acrodef{ESD}{energy spectral density}
\acrodef{FDD}{frequency division duplex}
    \acrodef{FDMA}{frequency division multiple access}
   \acrodef{GP}{Gaussian process}
    \acrodef{GPS}{global positioning system}
\acrodef{HetNet}{heterogeneous network}
    \acrodef{ICI}{inter-cell interference}
		\acrodef{IMI}{inter-mode interference}
\acrodef{LTE}{long term evolution}

\acrodef{MAC}{media access control}
\acrodef{MRU}{minimum resource unit}
\acrodef{MBB}{mobile broadband}
\acrodef{MTC}{machine type communications}
\acrodef{MCC}{mission critical communications}
\acrodef{MMC}{massive machine type communications}
   \acrodef{OFDM}{orthogonal frequency division multiplexing}
    \acrodef{PDF}{probability density function}
    \acrodef{PHY}{physical layer}
		\acrodef{PSD}{power spectral density}
		\acrodef{PP}{partition problem}
    \acrodef{PRB}{physical resource block}
   \acrodef{QoE}{quality of experience}
    \acrodef{QoS}{quality of service}
    \acrodef{RAN}{radio access network}
		\acrodef{RBS}{removal of bottleneck services}
		\acrodef{RMDI}{resource muting for dominant interferer}
    \acrodef{RRM}{radio resource management}
		\acrodef{RU}{resource unit}
		\acrodef{RX}{receiver}
 \acrodef{STCA}{scalable-TTI enabled channel allocation}
    \acrodef{SDN}{software defined network}
    \acrodef{SNR}{signal-to-noise ratio}
    \acrodef{SINR}{signal-to-interference-plus-noise ratio}
\acrodef{SIR}{signal-to-interference ratio}
\acrodef{SIF}{standard interference function}
    \acrodef{SVM}{support vector machine}
		
    \acrodef{TCP}{transmission control protocol}
		\acrodef{TDD}{time division duplex}
    \acrodef{TDMA}{time division multiple access}
		\acrodef{TTI}{transmission time interval}
		\acrodef{TX}{transmitter}
		\acrodef{UE}{user equipment}
		\acrodef{UL}{uplink}
    \acrodef{WLAN}{wireless local area network}
\usepackage{tabularx}

\IEEEoverridecommandlockouts                              

\begin{document}

\title{
	 An Examination of the Benefits of Scalable TTI for Heterogeneous Traffic Management in 5G Networks
}

\author{
	 	\IEEEauthorblockN{Emmanouil Fountoulakis\IEEEauthorrefmark{1},
		Nikolaos Pappas\IEEEauthorrefmark{1},
		Qi Liao\IEEEauthorrefmark{2},
	    Vinay Suryaprakash\IEEEauthorrefmark{2},
        Di Yuan\IEEEauthorrefmark{1}}\\
	\IEEEauthorblockA{\IEEEauthorrefmark{1} Department of Science and Technology, Link{\"o}ping University, Sweden}\\
	\IEEEauthorblockA{\IEEEauthorrefmark{2} Nokia Bell Labs, Stuttgart, Germany}\\
	E--mails: \{emmanouil.fountoulakis, nikolaos.pappas, di.yuan\}@liu.se\\
      	   \centering \{qi.liao, vinay.suryaprakash\}@nokia-bell-labs.com
             }
             	\maketitle
   
	\begin{abstract}
	 \textbf{The rapid growth in the number and variety of connected devices requires 5G wireless systems to cope with a very heterogeneous traffic mix. As a consequence, the use of a fixed transmission time interval  (TTI) during transmission is not necessarily the most efficacious method when  heterogeneous traffic types need to be simultaneously serviced.
	 This work analyzes  the benefits of scheduling based on exploiting scalable TTI, where the channel assignment and the TTI duration are adapted to the deadlines and requirements of different services. We formulate an optimization problem by taking individual service requirements into consideration. We then prove that the optimization problem is NP-hard and provide a heuristic algorithm, which provides an effective solution to the problem. Numerical results show that our proposed algorithm is capable of finding near-optimal solutions to meet the latency requirements of mission critical communication services, while providing a good throughput performance  for mobile broadband services.}
	\end{abstract} 
    
    \vspace*{2mm}
	
\section{Introduction}

The statement, ``Future wireless access will extend beyond people, to support connectivity for anything that may benefit from being connected.'', by the authors of \cite{EricssonReport} has far reaching implications. This entails that a  variety of  new autonomous devices, such as drones, sensors, etc., will communicate using the same network  that simultaneously has to serve  conventional mobile broadband (MBB) services. 
Thus,  next generation wireless communications systems will be characterized by their service requirement heterogeneity \cite{5Gwhitepaper}. 
A characteristic  example of  services, which have  requirements vastly different from   MBB services, are those that fall under the category of machine type  communications (MTC) \cite{FlexiblePedersen}. Two subcategories of MTC services are the mission critical communications (MCC) and  the massive  machine type communications (MMC).  MCC  services are characterized by small packets and  require ultra low latency ($\leq1$ms, \cite{EricssonReport}) and high reliability \cite{durisi2016toward}. On the other hand, MMC envisions tens of billions of connected devices \cite{EricssonReport}. Therefore, it is not far-fetched to assume that the use of a fixed TTI length for catering to such a diverse set of services could be suboptimal. For traffic types in which the ratio between the size of signaling and data is greater than or equal to  $1$, fixed TTI leads to a significant wastage of resources and -- as a result -- inefficient communications. The promise of scalable TTI as a potential solution was demonstrated in\cite{Qi}, where the TTI length could be scaled according to the traffic type.

To support a mix of  services with heterogeneous  requirements, in \cite{FlexiblePedersen} and \cite{FlexiblePedersen1}   the authors propose a flexible frame structure in frequency division dublex (FDD) networks. In these works, the delay constraints are reverse engineered based on the channel state information and the delay budgets. Along similar lines, the authors in \cite{levanen2014mm} apply the variable frame structure in the context of millimeter wave communications. However, these works aim to prioritize active services with strict latency requirements, while sacrificing the throughput of  mobile broadband users.  In a recent work \cite{Qi}, scalable TTI lengths are introduced in dynamic  time division duplex (TDD) mode in order to consider the requirements of each individual service and provide a good trade-off between heterogeneous performance metrics (with respect to their corresponding traffic demands and latency requirements). Moreover, the dynamic TDD scheme offers greater flexibility than the FDD scheme, in terms of adaptability to an asymmetry in  UL and DL traffic.
However, none of the works mentioned above jointly considers dynamic TTI length adaptation and channel allocation. In addition to  scheduling flexibility in the time domain, jointly considering scalable TTI and channel allocation provides a more flexible frame structure, which is better at exploiting channel diversity and improving spectral efficiency.

In this paper, we aim to develop a scheduling approach that strives to fulfill the (service) deadlines and requirements of different types of services by scaling the length of the TTI to be used. 
To this end, we formulate an optimization problem whose solution provides the appropriate TTI length and the channel allocation for each service.
We then prove that the optimization problem formulated is NP-hard. Therefore, in order to have a scheduler that works in polynomial time, we propose a greedy algorithm that finds an approximate  solution to the optimization problem. 
Numerical results show that the formulated optimization problem tries to cater to all MCC services within their latency requirements, while providing a higher throughput for MBB services in comparison to the other methods commonly considered. They also indicate that the improvement in performance provided by our formulation over the shortest deadline first scheduler (SDFS) increases as the number of active MCC services increases.
 
\section{System Model} \label{sec:model}

We consider a single cell of an FDD network   in  downlink mode \footnote{In this work, we assume that the downlink resources are always available since we consider an FDD system. However, the same formulation can also be applied to a TDD system, depending on whether the carriers are configured in uplink or downlink mode  during a given time period.}.
We also consider services, each with a deadline within which all their requirements must be met. Henceforth, we will use the term  services rather than users in recognition of the fact that a user can request more than one service. 
In this paper, we assume  discretized  time and \textquoteleft\textit{one time unit}' refers to the minimum amount  of time during which a transmission can occur. Let the TTIs be  indexed in the time domain by $n\in\field{N}$. The length of each TTI $\Delta(n), \forall n\in\NN$ is scalable and can be selected from a finite set $\Delta(n)\in\left\{1, 2, \ldots, L\right\}$, where $L \in\NN$ is the largest number of  time units  that can be assigned to a particular TTI.
The active set of services at the beginning of the $n$-th TTI is denoted by $\mathcal{S}_{n}$ with cardinality $|\mathcal{S}_{n}|$. 
	  
Let $\mathcal{K}\triangleq\left\{1,2,\ldots,K\right\}\subset\mathbb{N}$ be the set of available channels with cardinality $|\mathcal{K}|$, and assume that the same TTI size is retained for all the channels. Each service $s \in \mathcal{S}_{n}$ can be allocated to a number of channels. We use the vector $\bm{a}_{s}(n) \in \left\{0,1\right\}^{|\mathcal{K}|}$ to denote the allocation of  channels to a service $s$. The $i$-th element of $\bm{a}_{s}(n)$, $a_{i,s}(n)$, takes the value one if the $i$-th channel is assigned to the service $s$ during the $n$-th TTI, and takes the value zero otherwise. Let $\mathcal{NZ}_{s}(n)$ denote the set of non-zero elements of  vector $\bm{a}_{s}(n)$. Let the channel allocation for all services be collected in a binary matrix $\ma{A}(n)\in\{0,1\}^{|\mathcal{K}| \times |\mathcal{S}_n|}$, where the $s$-th column is $\ve{a}_s(n)$. Each channel can be assigned up to one service within a TTI and thus, we have the following constraint
	\begin{equation}
		\sum\limits_{s\in \mathcal{S}_{n}} a_{i,s}(n) \leq 1 , \forall  i \in \mathcal{K}, \, \forall n\in\NN. 
	\end{equation}
Each channel $i$ has a known channel state information (CSI) for every service $s$. The CSI in the $i$-th channel for the $s$-th service in the $n$-th TTI is a tuple defined as 
$$\csi_{i,s}\left(n\right)=\left(R_{i,s}(n),T_{i,s}(n)\right).$$ 
In this tuple, $R_{i,s}$ denotes the transmission rate of the $s$-th service over the $i$-th channel (in bits/one time unit) that can be sustained without errors for $T_{i,s}$ time units, if the $i$-th channel is assigned to $s$. Note that the CSI of a channel still changes from one TTI to another. 
	
At the beginning of the $n$-th TTI, each service $s$ has  a known data requirement  denoted by $Q_{s}(n-1)$. Then, we denote $Q_{s}(n)$   as the  amount of data (in bits) that still needs  to be served at the end of the $n$-th TTI. The evolution of the backlog can be described by
	\begin{equation}
        \hspace*{-3mm}Q_{s}(n)
		\triangleq \left[ Q_{s}(n-1)-\left(\Delta(n)-\delta \right)\sum\limits_{i\in \mathcal{K}}a_{i,s}(n)R_{i,s}(n)\right]^{+}\hspace*{-3mm},
	\end{equation}
where $\left[\cdot \right]^{+}  \triangleq \max\left\{0,\cdot\right\}$ and $\delta$ is the fraction of a time unit required for the transmission of the signaling overhead. We assume that $\delta$ is less than or equal to one time unit. Moreover, each service has a specific deadline before which the data has to be delivered. If a service is not completely served before the deadline, the system fails to meet its requirements and the service is dropped. This deadline is denoted by $D_{s}(n)$, and defined as
	\begin{equation}
	D_{s}(n)\triangleq \left[ D_{s}(n-1)-\Delta(n)\right]^{+}.
	\end{equation}
If $Q_{s}\left( n\right)\neq 0$ and $D_{s}\left( n\right)=0$, the service $s$ is dropped from the system, whereas if $Q_{s}\left(n\right) = 0$ and $D_s(n)\geq 0$, the service $s$ is completely served and exits the system. Additionally, we define the ``emptying rate'', $E_s(n)$, of a service $s$ at the end of the $n$-th TTI by
	\begin{equation}\label{EmptyingRate}
		E_s(n)\triangleq \frac{Q_{s}(n-1)-Q_{s}(n)}{Q_{s}(n-1)},  
	\end{equation}
where $E_s(n)\in \left[0,1\right]$, represents the ratio between the data served within the $n$-th TTI and the amount of data remaining at the end of the $(n-1)$-th TTI. This implies: the larger the emptying rate, the faster the data is served with respect to what was remaining at the end of the previous TTI. For example, if  service $s$ is completely served at the end of the third TTI, then $Q_s(3)=0$ and $E_s(3)=1$; on the other hand, if $s$ is not served at all during the third TTI, then $Q_s(2)=Q_s(3)$ and thus, $E_s(3)=0$.

\section{Problem Formulation}\label{sec:formulation}
	
At the $n$-th TTI, the optimization variables for the TTI length and the channel allocation are $\left\{\Delta(n),\mathbf{A}(n)\right\}$, respectively. Our objective is to address the trade-off between the throughput performance and number of dropped services. To this end, we develop a scheduling scheme  that will be able to either prioritize services with short deadlines, or(/and) services that can be completely served during the current round of scheduling.
	
\subsection{Utility function}
We define our utility function as
	\begin{equation}\label{eq:utility}
		U(n)\triangleq \sum\limits_{s\in\mathcal{S}_{n}} W_{s}(n)E_s(n),
	\end{equation}   
where $E_{s}(n)$ is the emptying rate, and the weight $W_{s}\triangleq\frac{1}{D_s(n-1)}$. Note that $W_s$ increases when the $D_s(n-1)$ decreases, i.e., its value increases if the deadline is soon to expire. Since we consider discrete time, the smallest value $D_s(n-1)$ can attain is one time unit. Therefore, the maximum value of $W_s$ is one and as a result, the maximum value of function $U(n)$ is equal to $|\mathcal{S}_n|$. Hence, the function provides a higher reward when the following  types of services are served: i) those having urgent deadlines; and, ii)  those that can be served with higher emptying rates.
	   
\subsection{Optimization Problem} 
Although the utility $U(n)$ in (\ref{eq:utility}) is designed to prioritize services with urgent deadlines, $U(n)$ alone cannot guarantee that services, which can be completely served during the current round of scheduling are chosen. Therefore, we formulate the optimization problem by augmenting the utility function and by introducing additional constraints, as given below.
\begin{subequations}
		\begin{align}\label{eq:util_opt}
		\underset{\Delta(n), \ma{A}(n)}{\max}  
		\quad &U(n) + \theta(n)\\ \label{eq:minTTI}  
		\text{s. t.}\quad & \Delta(n)\leq \min\limits_{s\in\mathcal{S}_{n}} \,\, \min\limits_{i\in \mathcal{NZ}_{s}(n)}T_{i,s}(n),\\ \label{eq:1servicechannel}
		& \sum\limits_{s\in \mathcal{S}_{n}} a_{i,s}(n) \leq 1,  \forall i \in \mathcal{K},\\ \label{set:TTI}
		& \Delta\left(n\right) \in \left\{1,\ldots, L\right\},\\
		& \ma{A}(n)\in\{0,1\}^{|\mathcal{K}| \times |\mathcal{S}_n|}, \label{eq:const_a}\\ \label{maxTTI}
		& \theta(n)= M\sum\limits_{s\in \mathcal{S}_{n}} \indication{Q_s(n)=0},  
		\end{align}  
	\end{subequations} 
where $M=(|\mathcal{S}_n|-1)$. Moreover, $\mathbb{1}_{\left\{B\right\}}$ is the indicator function which takes the value one if  the event $B$ occurs, and the value zero otherwise. For the rest of this paper, we refer to the problem above as  scalable-TTI enabled channel allocation STCA. The objective function (\ref{eq:util_opt}) is the  sum of the utility function (\ref{eq:utility}) and an additional reward $\theta(n)$. The function $\theta(n)$, defined in (\ref{maxTTI}), is equal to the product of a constant $M$ and the number of completely satisfied services at the end of the current TTI. This, therefore, ensures that the number of completely served services is included in the objective function (\ref{eq:util_opt}). Furthermore, $\theta(n)$ also ensures that if  at least one service is completely served, the value it takes in the corresponding term of the objective function (\ref{eq:util_opt}) is greater than the sum of the other $(|\mathcal{S}_n|-1)$ terms of the objective function. As a result, we prioritize services that can be completely served after the current scheduling instance.       
    
Additionally, constraint (\ref{eq:minTTI}) ensures that the selected TTI size does not violate the minimum TTI size for a given channel and service. Constraint (\ref{eq:1servicechannel}) ensures that a channel can be assigned to up to one service. 
  
\section{Complexity}\label{Sec:Compl}

This section addresses the complexity of the optimization problem. Specifically, we  prove that the optimization problem, as defined in Section \ref{sec:formulation}, is NP-hard. However, as shown later on in Theorem \ref{the2}, the problem admits a polynomial-time algorithm  guaranteeing  optimality,  if flat channels are assumed. By flat channels, we mean that for each service, the channel gains are the same for all channels within a given TTI.
	
\begin{thm}\label{Th:NP}
	STCA is NP-hard. 
\end{thm}
\begin{proof}
We prove that the decision version of the STCA  problem is NP-complete by a polynomial-time reduction to and from the Partition Problem (PP) in three steps, \cite{Computers}. The decision version of the STCA problem can be stated as:

\textit{ Given a set of services $\mathcal{S}_n$, the backlogs $Q_s(n-1)$,  the deadlines $D_s(n-1)$,  a set of channels $\mathcal{K}$, and the achievable rates $R_{i,s}(n)$, $\forall i\in \mathcal{K}$ and $\forall s\in\mathcal{S}_n$, is there a solution of the given STCA instance such that the  value of the objective function is at least $f$, where $f$ is a given positive number?}
		
\noindent \textit{\underline{Step 1}}: We prove that the STCA problem belongs to the NP class of problems, i.e. given an STCA instance, a positive answer and its associated solution, it takes polynomial time to verify whether the answer to the question posed is indeed YES. It is a plain to see that, given a solution,  computing $U(n) + \theta(n)$ takes polynomial time. Therefore, STCA is in the NP class of problems. 
		
\noindent \textit{\underline{Step 2}}: We now show that there is a polynomial-time reduction from the PP to the STCA problem. In the PP, for a set of positive integers $\left\{p_1,\ldots,p_m\right\}$, the task is to determine whether or not this set can be partitioned into two subsets of equal sums, i.e. $\sum\limits_{i\in \mathcal{A'}} p_{i}=\sum\limits_{i \in \mathcal{A\setminus A'}} p_{i}$, where $\mathcal{A}=\left\{1,\ldots,m\right\}$ and $\mathcal{A'} \subset \mathcal{A}$. Without loss of generality, we can assume that $\sum\limits_{i\in \mathcal{A}} p_i$ is even. Then, given an instance of the PP, we can define an instance of the STCA problem as follows:\vspace*{-1mm}
	\begin{itemize}
		\item $\mathcal{S}_{n}=\left\{1,2\right\}, \Longrightarrow |\mathcal{S}_{n}|=2$. $|\mathcal{K}|=|\mathcal{A}|$.
		\item $D_{s}(n-1)=1$ time unit, $\forall s\in \mathcal{S}_{n}$.
		\item $\Delta(n)=1$ time unit.
		\item $\delta=0$. $R_{i,s}(n)=p_i$, $\forall s \in \mathcal{S}_n, \forall i \in \mathcal{A}$.
		\item $Q_s(n)=\frac{1}{2}\sum\limits_{i \in \mathcal{A}} p_i$, $\forall s \in \mathcal{S}_n$.
	\end{itemize} 
Based on the instance defined above, the value of $f$ in the decision version of this STCA instance is set to $4$, i.e., $f=4$. From the assignments above, there is a one-to-one mapping between the elements in the PP and the channels in the STCA problem. In particular, we associate the $i$-th element in $\mathcal{A}$ with the $i$-th element in $\mathcal{K}$. Therefore, the above definition clearly represents  a polynomial-time reduction.
		     
\noindent
\textit{\underline{Step 3}}: We now prove that the PP instance has the answer YES if and only if the answer to the defined STCA decision  instance is YES. If the answer to the PP instance is YES, there are two sets $\mathcal{A'}$ and $\mathcal{A\setminus A'}$, such that $\sum\limits_{i\in \mathcal{A'}} p_i=\sum\limits_{i\in \mathcal{A\setminus A'}} p_i= \frac{1}{2} \sum\limits_{i \in \mathcal{A}} p_i$. We assign the channels corresponding to the set $\mathcal{A'}$ to  one service, and the channels corresponding to the set $\mathcal{A\setminus A'}$ to  the other. Hence, for the STCA instance, we have $\sum\limits_{i\in\mathcal{A'}} R_{i,1}=\sum\limits_{i\in\mathcal{A\setminus A'}} R_{i,2}=\frac{1}{2} \sum\limits_{i \in \mathcal{A}} p_i$. Since $Q_{s}(n)=\frac{1}{2} \sum\limits_{i\in \mathcal{A}} p_i \text{, }\forall s \in \mathcal{S}_{n}$, both services are completely served and therefore, $f=4$. Hence, the instance above is a YES instance of the defined STCA decision problem.
		 
Conversely, if the answer to the defined STCA decision instance is YES, there are two sets $\mathcal{K'}$ and $\mathcal{K\setminus K'}$ which correspond to the channel assignments for the services one and two, respectively. Since the answer is YES, there is a solution such that the value of the objective function is equal to $4$.  Note that this value can be reached if and only if both  services are completely served. Hence, we have	\begin{align}\label{rate1:geq}
		\sum\limits_{i \in \mathcal{K}'} R_{i,1}(n)& \geq \frac{1}{2} \sum\limits_{i \in \mathcal{A}}p_i, \\ \label{rate2:geq}
		\sum\limits_{i \in \mathcal{K} \setminus \mathcal{K}'}\hspace*{-2mm} R_{i,2}(n)& \geq \frac{1}{2} \sum\limits_{i \in \mathcal{A}}p_i.
	\end{align}
We also have, by definition, that $\sum\limits_{i \in \mathcal{K}} R_{i,s}(n) = \sum\limits_{i \in \mathcal{A}} p_i,$ for $s \in \{1,2\}$, and $R_{i,1}(n)=R_{i,2}(n)=p_{i}\text{, } \forall i$. Therefore, the conditions (\ref{rate1:geq}) and (\ref{rate2:geq}) hold if and only if they are equal. Hence, $\sum\limits_{i \in \mathcal{K'}} p_i= \sum\limits_{i \in \mathcal{K\setminus K'}} p_i =\frac{1}{2} \sum\limits_{i \in \mathcal{K}} p_i$, and $\{\mathcal{K},\mathcal{K\setminus K'}\}$ is a feasible partition. This establishes the NP-completeness of the decision version of the STCA problem. Therefore, the STCA problem is NP-hard. 
\end{proof}     
	 
This leads us to the proof that the global optimum of \textit{STCA} can be computed in polynomial time for the special case of flat channels.
\begin{thm}\label{the2}
	The global optimum of \textit{STCA} can be computed in polynomial time for flat channels.
\end{thm}
\begin{proof}
	If we have $K$ flat channels, then  $\csi_{k,s_i} = \csi_{l,s_j}$, for all channels $k$ and  $l$, and for all services $s_i$ and $s_j$. Let $g_k^s$ denote the value of the objective function when $k$ channels are allocated to service $s$, i.e.
	\begin{equation}  
		g_{k}^{s}=
		\begin{cases}
		W_{s}(n)+M, & \text{if } Q_s(n)=0 \equiv E_{s}(n)=1,\\
		W_{s}(n) E_s(n), & \text{otherwise}.	 
		\end{cases}
	\end{equation}
	Moreover, if there is no channel assigned to the service $s$, then $g_{0}^{s}=0$. Let $h_s(i)$ denote the objective function value of optimally allocating $i$ channels to services $\left\{1, \ldots, s\right\}$. The optimal objective value can be computed by the  recursive function
	\begin{equation}\label{eq:h_sk}
		h_{s}(k)=\max\limits_{k=0,1,\ldots,K}\left\{g_{k}^{s}+h_{s-1}(K-k)\right\}\text{.}
	\end{equation}
	We then construct a $|\mathcal{S}_{n}|\times K$ matrix whose elements are computed using (\ref{eq:h_sk}). The $(s,k)$-th element of the matrix includes the optimal value of the objective function for services $\left\{1, \ldots, s\right\}$ using $k$ channels. Hence, the $(|\mathcal{S}_{n}|,K)$-th element gives the value of the optimum solution of the entire optimization problem.
		
	For the first row of the matrix, computing the entries $h_{1}\left(1\right),\ldots,h_{1}\left(k\right)$ in the given order are straightforward, and each entry requires a computational complexity of $\mathcal{O}\left(1\right)$. Each element of the $s$-th row requires $\sum\limits_{i=1}^{K} i = K\left(K+1\right)/2$ computations. Hence, the computational complexity that is required for each row is $\mathcal{O}\left(K^{2}\right)$ and thus, the total computational complexity is $\mathcal{O}\left(|\mathcal{S}_{n}|K^{2}\right)$. Therefore, the optimum solution of the STCA problem, in the case of flat channels, can be computed using dynamic programming in polynomial time.
	\end{proof}	
	 
\section{Integer Linear  Programing Formulation}\label{sec:IP} 
 
In this section, we develop an Integer Linear Program (ILP) in order to compute the optimal solution of the STCA problem, which enables a more detailed study of the performance of scalable TTI. First, we solve the problem in (\ref{eq:util_opt}) with a fixed TTI length as an input. Note that the problem is solved for each viable TTI length separately. Then, we compare the value of the objective function for all the TTI lengths considered, and  subsequently select the TTI length and the channel assignment for which the objective function is maximized. The pair $\left\{\Delta(n),\bm{A}(n)\right\}$ for which the objective function in (\ref{eq:util_opt}) is maximized is the optimal solution. It should be noted that, for each possible TTI length, if the TTI length is greater than a given service's deadline, we remove the corresponding service from the optimization problem; thereby, considering the service dropped. In other words, the services whose deadlines will expire despite choosing the optimal $\Delta$ (denoted by $\Delta'$) have a utility equal to zero. Thus, for each fixed $\Delta'$, we consider the set of services $\{s\in\mathcal{S}_n:D_s(n-1)\geq \Delta'\}$.
	
In this section, we omit the index $n$ for notational brevity and redefine some of the parameters as follows:
	\begin{itemize} 
		\item $Q'_{s}$ -- the data backlog of $s$ during the current TTI.
	    \item $W'_{s}=\frac{W_{s}}{Q'_s}$.
		\item $\beta_{s}$ -- amount of data served to the service $s$ at the  end of the current TTI.
		\item $R'_{i,s} = (\Delta - \delta) R_{i,s}$ is the amount of data that could be transmitted to service $s$, if the channel $i$ is assigned to it. 
		\item $Y_s=
		         \begin{cases}
		                   1\text{, if the service } s \text{ is completely served,}\\
		                   0\text{, otherwise.}
		         \end{cases}$
		 \item $D_s'$-- the deadline of service $s$ after the $(n-1)$-th TTI.
		 \item $\mathcal{S}_{\Delta'}=\{s\in\mathcal{S}_{n}:D_s'\geq \Delta'\}$.
	\end{itemize}
The rest of the notations remain unchanged. The optimization problem can then be formulated  as the following ILP for a given $W_{s} \in \RN$ and  $\Delta'$.
	\begin{subequations}
		\begin{eqnarray}\label{obj:integer}
		&{\max\limits_{\mathbf{A}}}  \sum\limits_{s \in \mathcal{S}_{\Delta'}}W'_s\beta_{s} + M\sum\limits_{s\in\mathcal{S}_{\Delta'}}Y_s \\[5pt]  \label{eqn:cnst_TTI}
		\text{s.~t.} &\Delta' - T_{i,s}\leq J_1(1-a_{i,s}), \forall i\in\mathcal{K}, \forall s\in\mathcal{S}_{\Delta'}, \label{eqn:cnst_TTI}\\[5pt] \label{chanassign:integer}
		& \sum\limits_{s\in \mathcal{S}_{\Delta'}}a_{i,s}\leq 1,  \forall i\in \mathcal{K},\\[5pt] \label{consb:serveddata}
		& \beta_s\leq \sum\limits_{i \in \mathcal{K}} R'_{i,s}a_{i,s} \label{consb:upperdata},  \forall s \in \mathcal{S}_{\Delta'},  \\[5pt] 
        &\label{const:Ys} Y_{s}\leq \frac{\beta_s}{Q'_s}\leq 1, \forall s\in \mathcal{S}_{\Delta'},
		\end{eqnarray}  
	\end{subequations}
where the constant $J_1\gg L$ in \eqref{eqn:cnst_TTI} guarantees that $a_{i,s} = 0$ if $T_{i,s}<\Delta'$. The constraint (\ref{chanassign:integer}) ensures that each channel is assigned  up to one service and (\ref{consb:upperdata}) makes sure that the maximum value $\beta_s$ can attain is the amount of data remaining  for service $s$. Therefore, if the service $s$ is completely served, the corresponding term in (\ref{obj:integer}) takes the maximum value, which is equal to $W_s$. Note that the ratio $\frac{\beta_{s}}{Q'_s}$ in (\ref{const:Ys}) represents the emptying rate in (\ref{EmptyingRate}). Additionally, if $s$ is completely served, constraint (\ref{const:Ys}) ensures that $Y_s$ is assigned a value equal to one.

\section{Algorithm}
\setlength{\textfloatsep}{0pt}
\begin{algorithm}[!h]
   	\scriptsize
   	\caption{CAST algorithm}
   	$G_{\text{max}}\leftarrow - \infty$, $W_{s}=\frac{1}{D_{s}(n-1)}\text{, }\forall s \in\mathcal{S}$\\
   	\For{$\Delta'=1:L$}{
   	      $\mathbf{A}'\leftarrow \mathbf{0}_{K\times |\mathcal{S}|}\text{, }\mathcal{S}'\leftarrow\mathcal{S}\text{, } Q'_{s}\leftarrow Q_{s}$\\
           \If{$D_{s}(n-1)-\Delta'<0$}{
              $\mathcal{S}'\leftarrow\mathcal{S}\setminus\{s\}$    
           }
           \For{$i\in \mathcal{K}$}{
           	$g_{\text{max}}\leftarrow -\infty$\\
           	     \For{$s\in \mathcal{S}'$}{
           	     	    \eIf{$\Delta'\leq T_{i,s}$}{
           	     	          $Q_{\text{temp}}\leftarrow \left[Q'_s-(\Delta'-\delta)R_{i,s}\right]^{+}$ \\[2pt]
           	     	          $E'_{s}\leftarrow \frac{Q_{s}(n-1)-Q_{\text{temp}}}{Q_{s}(n-1)}$\\[2pt]
           	     	          $g\leftarrow W_{s}E_{s}'+M\mathbb{1}_{\left\{Q_{\text{temp}=0}\right\}}$\\
           	     	          \If{$g>g_{\text{max}}$}{
           	     	        	  $s_{\text{max}}\leftarrow s$, $g_{\text{max}}\leftarrow g$\\
           	     	       	      $Q_{s_\text{max}}\leftarrow Q_{\text{temp}}$} 
           	     	              \If{$Q_{s_\text{max}}= 0$}{
           	     	    	          $\mathcal{S'}\leftarrow \mathcal{S}\setminus \left\{s_{\text{max}}\right\}$
                	                  }
                        }
                        {$A'_{i,s}\leftarrow 0$}
           	     }
             $G\leftarrow G+g_{\text{max}}$, $A'_{i,s_\text{max}}\leftarrow 1$\\
           }
        \If{$G>G_{\text{max}}$}{
       	$\mathbf{A}_{\text{max}}\leftarrow \bf{A}'$\\
       	$\Delta_{\text{max}}\leftarrow \Delta'$
       } 
    }
	$\mathbf{A}(n)\leftarrow \mathbf{A}_{\text{max}}, \Delta(n)\leftarrow \Delta_{\text{max}}$
   \end{algorithm}

In order to have a polynomial time scheduling algorithm, we propose a heuristic  called  channel allocation with scalable TTI (CAST) algorithm.
For each channel $i \in \mathcal{K}$, the CAST algorithm finds the service $s \in \mathcal{S}_{n}$, which has the maximum corresponding value of the objective function (\ref{eq:util_opt}) -- should the channel $i$ be assigned to service $s$. The algorithm calculates the objective function for each possible TTI length, and selects the channel assignment and the TTI length for which the objective function is maximized.

The CAST algorithm decides the channel assignment for each  TTI length in two  steps. During the first step, the algorithm excludes the services whose deadlines cannot be met (lines 4 -- 5). The variable $g$, whose value is calculated in lines 9 -- 12, is the objective function value, if the channel $i$ is assigned to the service $s$. Note that a channel $i$ can be assigned to service $s$ only if the TTI length $\Delta'$ is less than the duration $T_{i,s}$ within which an error-free computation of the rate is possible (cf. line 9). During the second step, the algorithm allocates each channel to a corresponding service with the maximum value of the objective function (cf. lines 14 -- 15) and removes the service if it is completely served (lines 16 -- 17). The algorithm then compares the value of the objective function for each possible TTI length and selects the channel assignment as well as the TTI length maximizing the value of the objective function (lines 21 -- 24). Based on the description of ILP above, the complexity of the CAST algorithm is found to be  $\mathcal{O}\left(|\mathcal{K}||\mathcal{S}_{n}|L\right)$.

\section{Numerical Results} \label{sec:NumRes} 

In this section, we compare the performance of the CAST algorithm with the optimal solution (OS) for the STCA problem. Additionally, we also compare our approach with a simpler version of  the shortest deadline first scheduler (SDFS)  proposed by the authors in \cite{FlexiblePedersen1}. The above mentioned comparisons are undertaken using the simulations based on the parameters that follow. 

We consider  one time unit to be equal to $0.1$ms, and the TTI length can be selected from a finite set $\Delta(n) \in \left\{ 0.2\text{ms}, 0.3\text{ms}, \ldots, 1\text{ms}\right\}$ in a single cell scenario where the FDD is in downlink mode \footnote{\scriptsize Note that $\Delta(n)$ here is presented with the units 'milliseconds' for improved readability. The value of $\Delta(n)$ in milliseconds is obtained by multiplying the original $\Delta(n)$ with the duration of one time unit ($0.1$ms).}. We also assume that the transmission of control signaling  requires $\delta=0.05$ms  per TTI (regardless of the length of the TTI chosen). We consider a system with an $8$ MHz bandwidth that works on a frequency selective channel with a coherence bandwidth of $0.5$ MHz. The achievable rate for a service $s$ in the $i$-th channel during the $n$-th TTI is computed using the Shannon formula and is given by
$R_{i,s}(n)=B\log_{2}\left(1+|h_{i,s}(n)|^2\frac{S}{N}\right)$,
where  the channel gains $h_{i,s}(n)$ are distributed as a zero-mean complex Gaussian with variance $\sigma^2$, i.e., $h_{i,s}(n)\sim \mathcal{CN}\left(0,\sigma^2\right)$, $S$ is the transmit power, $N$ is the noise power, and $B$ is the  bandwidth of each channel, i.e., $B=0.5$ MHz. The average value of the signal-to-noise ratio (SNR) is equal to $5$ dB. Moreover, we consider that the base station caters to services generated by three MCC sources and one MBB source. Each source generates services per time unit ($0.1$ms) according to a Bernoulli  distribution with probability $r_{\text{MCC}}$ and $r_{\text{MBB}}$ for  MCC sources and the MBB source, respectively. Lastly, each MCC service has a demand of $125$ bytes and deadline of $1$ms, and each MBB service has a demand of $1125$ bytes and a deadline of $10$ms. In the following paragraphs, we study the behavior of the algorithms proposed for various values of $r_{\text{MCC}}$, while the probability of MBB service arrivals is constant and equal to 0.2, i.e., $r_{\text{MBB}} = 0.2$.

\begin{figure}[!htb]
		\vspace*{-2mm}
        \centering
		\includegraphics[scale=0.5]{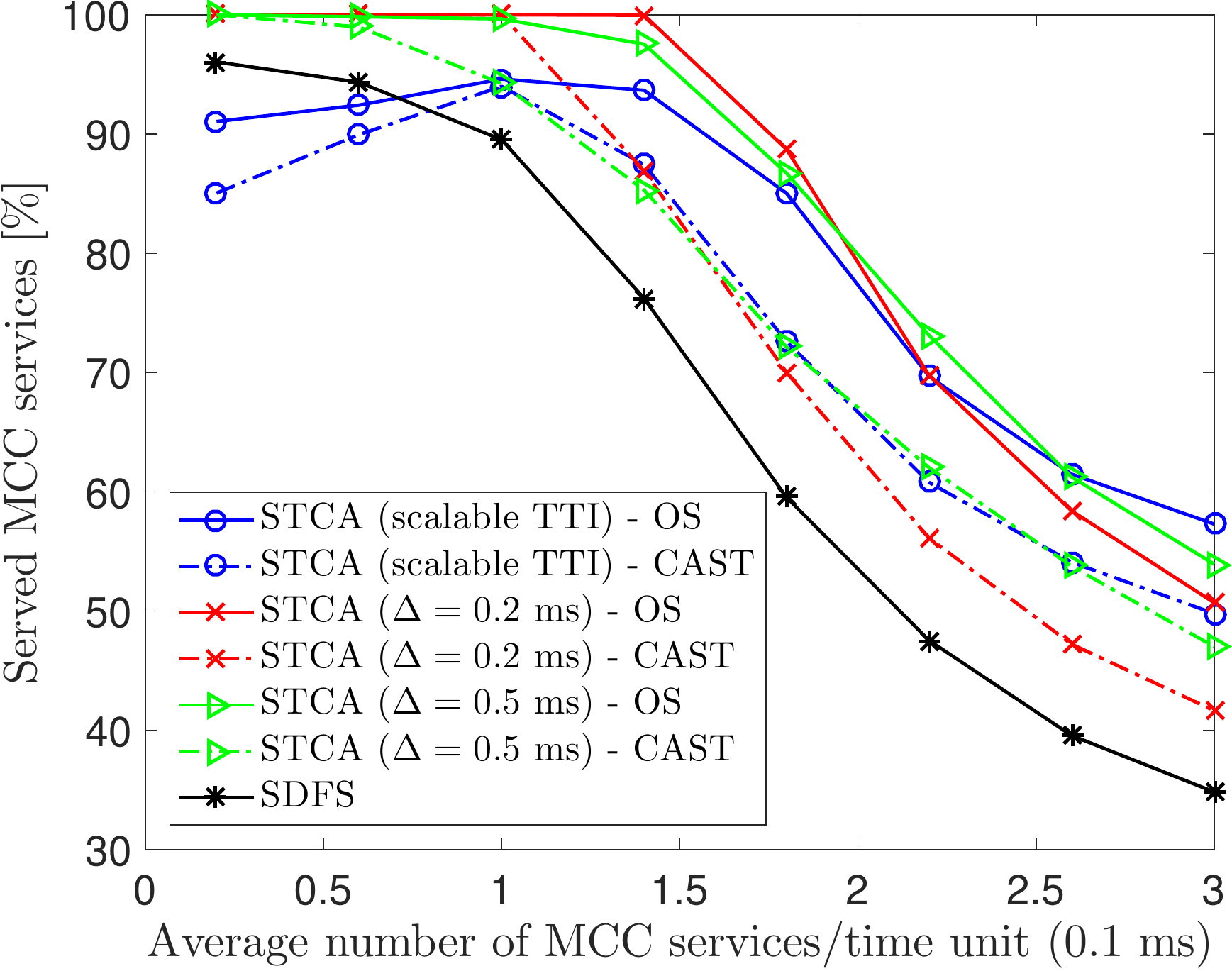}
		\caption{Variations in MCC services.} 
		\label{ServedMCC}
		\vspace*{-5mm}
\end{figure} 
\begin{figure}[!htb]
		\centering
		\includegraphics[scale=0.5]{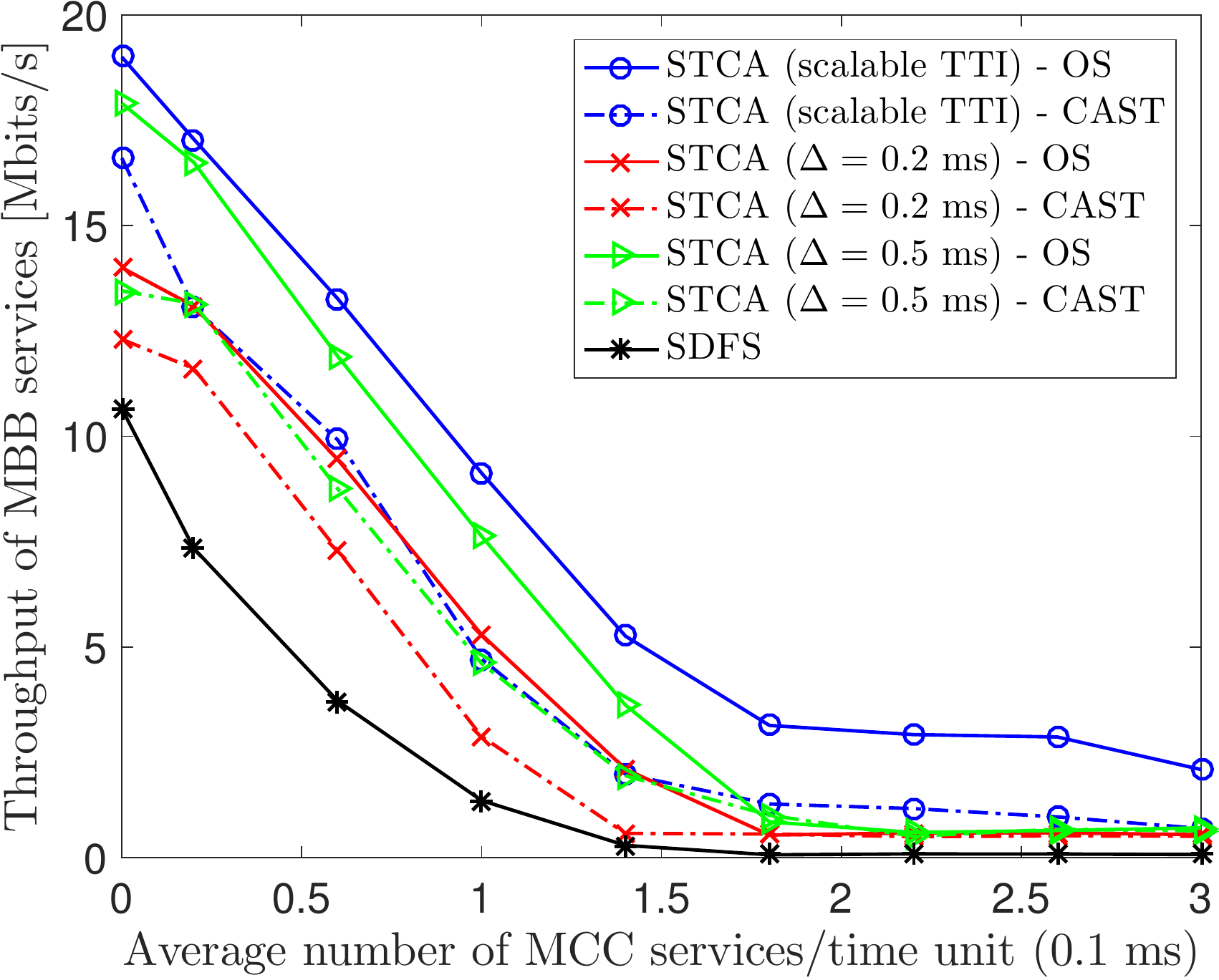} 
		\caption{Variations in the throughputs of MBB services.}
		\label{ThrMBB}
		\vspace*{-2mm}
\end{figure} 
\begin{figure}[!htb]
		\centering
		\includegraphics[scale=0.5]{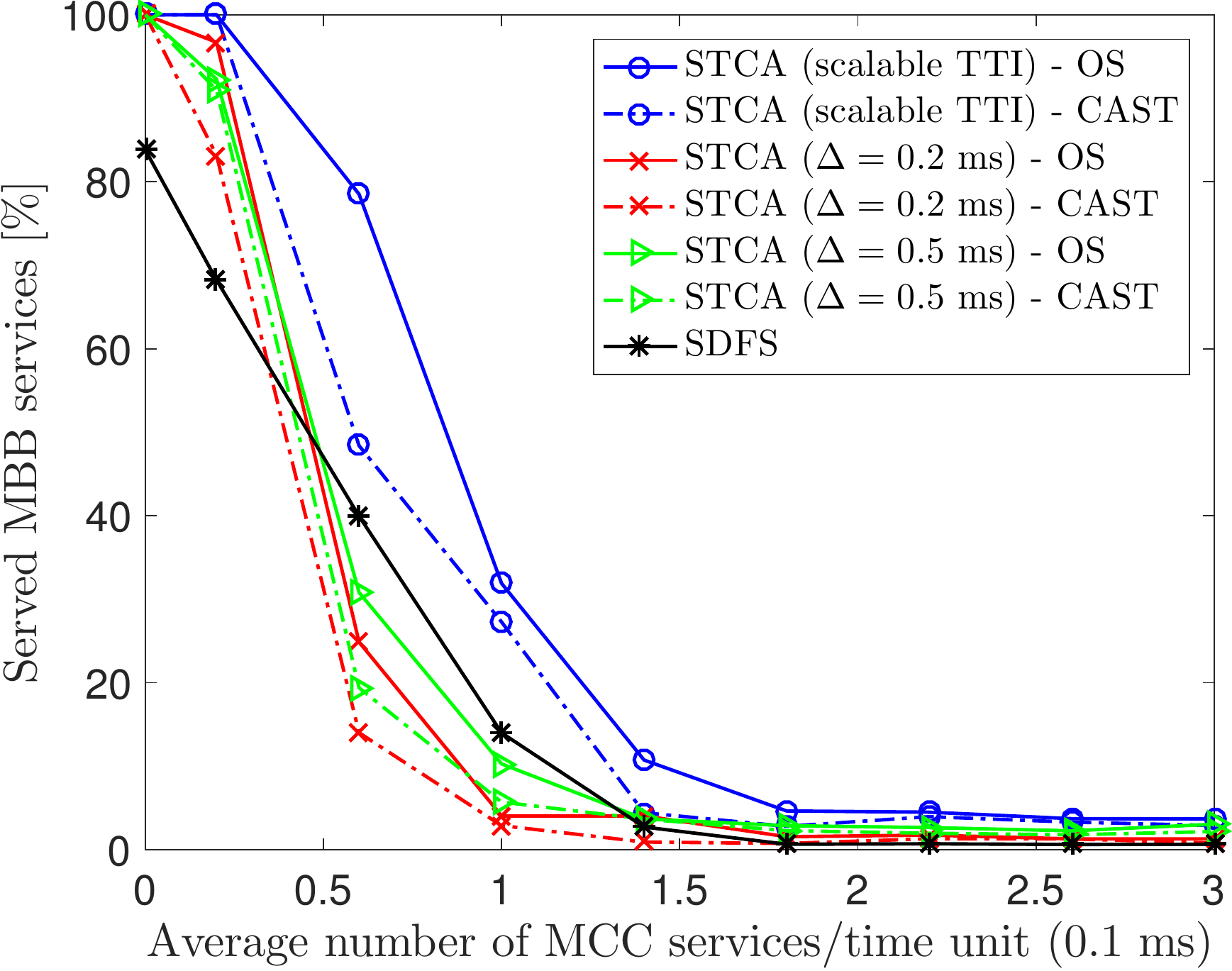}
		\caption{Variations in MBB services.}
		\label{ServedMBB}
\end{figure} 

Fig.~\ref{ServedMCC} depicts the variations in the percentage of MCC services dealt with as the average number of MCC service requests per time unit ($0.1$ms) increases. It documents the aforementioned variations for both the optimal solution and the heuristic of the STCA in scenarios where the TTI lengths are scalable and fixed, as well as the variations seen in the behavior of the SDFS. This figure indicates that a scheduler using the STCA with short but fixed TTI lengths outperforms the one using the STCA with scalable TTI as well as the SDFS. The reason why the STCA with short, fixed TTI outperforms the STCA with scalable TTI is because the latter tends to select longer TTI lengths in order to be able to \textit{completely} serve as many services as possible during each scheduling period. This sort of selection implies that a greater portion of the MCC services end up being dropped. However, as the arrival rate of MCC services continues to increase, the STCA with scalable TTI starts to select shorter TTI lengths; thereby, resulting in the increase in the percentage of MCC services catered to between $0.2$ and $1$ MCC arrivals$/ 0.1$ms before eventually decreasing beyond $1.5$ MCC services$/0.1$ms. It is noteworthy that the STCA with scalable TTI eventually outperforms the STCA with fixed TTI, i.e., beyond $2.5$ MCC services$/0.1$ms. 
 
As commonly known, the amount of signaling overhead increases quite substantially when shorter TTI lengths are selected. The cost of an increase in the signaling overhead is born a decrease in the throughput delivered to the MBB services. Fig.~\ref{ThrMBB} demonstrates the variations in the throughput of the MBB services as the average number of MCC service requests$/0.1$ms increases. Clearly, of the methods considered, the SDFS is the one that is most significantly affected. This figure also indicates that, though the MBB services see an inevitable drop in their throughput, the STCA with scalable TTI is able to cope much better than the STCA with short, fixed TTI -- especially when the average number of MCC service requests$/0.1$ms is greater than $1.5$. A reason why the STCA with scalable TTI outperforms the STCA with short, fixed TTI is because of its ability to contain (and regulate) the amount time spent in transmitting the control signaling more effectively. 
 
Lastly, Fig.~\ref{ServedMBB} -- as in Fig.~\ref{ThrMBB} -- depicts the unavoidable decrease in the percentage of MBB services satisfied when the average number of MCC service requests$/0.1$ms increases. It does, however, highlight the fact that the STCA with scalable TTI is able to serve a far greater percentage of MBB services when compared to the others in the face of increasing MCC service requests$/0.1$ms. This behavior can, once again, be attributed to the fact the STCA with scalable TTI can control the fraction of time spent transmitting the control signaling by periodically choosing larger TTI lengths and thereby, ensuring that MBB services are also furnished with the resources they need. Also, the results illustrate that there is a visible gap between the performance of the CAST algorithm and the OS, though the CAST algorithm significantly outperforms the SDFS. This gap is expected because of the low complexity of the CAST algorithm.

Overall, when one considers all the results collectively, it can be said that a scheduler which jointly considers scalable TTI and channel allocation into account is better at being able to handle traffic heterogeneity and has the ability to improve the spectral efficiency of individual service types.
 
\section{Conclusions}
In this paper, at each scheduling time, we propose a joint optimization of the TTI lengths and the channel allocation depending on the traffic type. The joint optimization problem formulated is then proven to be NP-hard due to which we provide a heuristic akin to a greedy algorithm. However, for flat channels, we also demonstrate that the problem admits a polynomial-time solution that guarantees optimality. 
The optimization problem and its heuristic are then compared not only with one another for the cases of fixed and scalable TTI lengths, but also with the shortest deadline first scheduler. These evaluations illustrate that our proposal of
a joint optimization of TTI lengths and channel allocation is better equipped to handle traffic heterogeneity and provide improved spectral efficiency, due to its ability to regulate the amount of time spent on control signal transmissions and maximize the number of services satisfied. 
  
\section{Acknowledgment}
The authors would like to thank Dr. Ilaria Malanchini for numerous fruitful discussions and her valuable suggestions.
This  work has been supported by the European
Union's Horizon 2020 research and innovation programme under
the Marie Sk\l{}odowska-Curie grant agreement No. $643002$.

\bibliographystyle{IEEEtran}
\bibliography{MyBib}


\end{document}